\documentclass{llncs}

\usepackage[dvips,letterpaper,margin=1.5in]{geometry}

\usepackage{epsfig, amssymb, amsmath}
\usepackage{pstricks, pst-node, pst-tree, color}
\usepackage{amsfonts}

\def\CG{{\cal G}}

\def\CSUM{{\mathrm{S}\textsc{um}}}

\def\RP{{\mathbb{R}_{\geq 0}}}

\begin{document}

\title{{A Unifying Tool for Bounding the Quality of Non-Cooperative Solutions in Weighted Congestion Games}\thanks{This work was partially supported by the PRIN 2008 research project COGENT ``Computational and game-theoretic aspects of uncoordinated networks'' funded by the Italian Ministry of University and Research.}}
\author{Vittorio Bil\`o\institute{Department of Mathematics, University of Salento, Provinciale Lecce-Arnesano, P.O. Box 193, 73100 Lecce, Italy. Email: {\tt vittorio.bilo@unisalento.it}}}

\maketitle

\begin{abstract}
We present a general technique, based on a primal-dual formulation, for analyzing the quality of self-emerging solutions in weighted congestion games. With respect to traditional combinatorial approaches, the primal-dual schema has at least three advantages: first, it provides an analytic tool which can always be used to prove tight upper bounds for all the cases in which we are able to characterize exactly the polyhedron of the solutions under analysis; secondly, in each such a case the complementary slackness conditions give us an hint on how to construct matching lower bounding instances; thirdly, proofs become simpler and easy to check. For the sake of exposition, we first apply our technique to the problems of bounding the prices of anarchy and stability of exact and approximate pure Nash equilibria, as well as the approximation ratio of the solutions achieved after a one-round walk starting from the empty strategy profile, in the case of affine latency functions and we show how all the known upper bounds for these measures (and some of their generalizations) can be easily reobtained under a unified approach. Then, we use the technique to attack the more challenging setting of polynomial latency functions. In particular, we obtain the first known upper bounds on the price of stability of pure Nash equilibria and on the approximation ratio of the solutions achieved after a one-round walk starting from the empty strategy profile for unweighted players in the cases of quadratic and cubic latency functions. We believe that our technique, thanks to its versatility, may prove to be a powerful tool also in several other applications.
\end{abstract}

\


\section{Introduction}
Characterizing the quality of self-emerging solutions in non-cooperative systems is one of the leading research direction in Algorithmic Game Theory. Given a game $\cal G$, a social function $\cal F$ measuring the quality of any solution which can be realized in $\cal G$, and the definition of a set $\cal E$ of certain self-emerging solutions, we are asked to bound the ratio ${\cal Q}({\cal G},{\cal E},{\cal F}):={\cal F}(K)/{\cal F}(O)$, where $K$ is some solution in ${\cal E}({\cal G})$ (usually either the worst or the best one with respect to $\cal F$) and $O$ is the solution optimizing $\cal F$.

In the last ten years, there has been a flourishing of contribution in this topic and, after a first flood of unrelated results, coming as a direct consequence of the fresh intellectual excitement caused by the affirmation of this new research direction, a novel approach, aimed at developing a more mature understanding of which is the big picture standing behind these problems and their solutions, is now arising.

In such a setting, Roughgarden \cite{R09} proposes the so-called ``smoothness argument" as a unifying technique for proving tight upper bounds on ${\cal Q}({\cal G},{\cal E},{\cal F})$ for several notions of self-emerging solutions $\cal E$, when $\cal G$ satisfies some general properties, $K$ is the worst solution in $\cal E(\cal G)$ and $\cal F$ is defined as the sum of the players' payoffs. He also gives a more refined interpretation of this argument and stresses also its intrinsic limitations, in a subsequent work with Nadav \cite{NR10}, by means of a primal-dual characterization which shares lot of similarities with the primal-dual framework we provide in this paper. Anyway, there is a subtle, yet substantial, difference between the two approaches and we believe that the one we propose is more general and powerful. Both techniques formulate the problem of bounding ${\cal Q}({\cal G},{\cal E},{\cal F})$ via a (primal) linear program and, then, an upper bound is achieved by providing a feasible solution for the related dual program. But, while in \cite{NR10} the variables defining the primal formulation are yielded by the strategic choices of the players in both $K$ and $O$ (as one would expect), in our technique the variables are the parameters defining the players' payoffs in $\cal G$, while $K$ and $O$ play the role of fixed constants. As it will be clarified later, such an approach, although preserving the same degree of generality, applies to a broader class of games and allows for a simple analysis facilitating the proof of tight results. In fact, as already pointed out in \cite{NR10}, the Strong Duality Theorem assures that each primal-dual framework can always be used to derive the exact value of ${\cal Q}({\cal G},{\cal E},{\cal F})$ provided that, for any solution $S$ which can be realized in $\cal G$, ${\cal F}(S)$ can be expressed though linear programming and
\begin{itemize}
\item[$\bullet$] the polyhedron defining $\cal E(\cal G)$ can be expressed though linear programming, when $K$ is the worst solution in $\cal E(\cal G)$ with respect to $\cal F$,
\item[$\bullet$] the polyhedron defining $K$ can be expressed though linear programming, when $K$ is the best solution in $\cal E(\cal G)$ with respect to $\cal F$.
\end{itemize}
Moreover, in all such cases, by applying the ``complementary slackness conditions", we can figure out which pairs of solutions $(K,O)$ yield the exact value of ${\cal Q}({\cal G},{\cal E},{\cal F})$, thus being able to construct quite systematically matching lower bounding instances.

We consider three sets of solutions $\cal E$, namely,
\begin{itemize}
\item[$\bullet$] pure Nash equilibria (PNE), that is, the set of outcomes in which no player can improve her situation by unilaterally changing the adopted strategy;
\item[$\bullet$] $\epsilon$-approximate pure Nash equilibria ($\epsilon$-PNE), that is, outcomes in which no player can improve her situation of more than a factor $\epsilon$ by unilaterally changing the adopted strategy (by definition, each $0$-PNE is a PNE);
\item[$\bullet$] solutions achieved after a one-round walk starting from the empty strategy profile \cite{MV04}, that is, the set of outcomes which arise when, starting from an initial configuration in which no player has done any strategic choice yet, each player is asked to select, sequentially and according to a given ordering, her best possible current strategy.
\end{itemize}
For $\cal E(\cal G)$ being the set of $\epsilon$-PNE of $\cal G$, ${\cal Q}({\cal G},{\cal E},{\cal F})$ is called the approximate price of anarchy of $\cal G$ ($\epsilon$-PoA$(\cal G)$) when $K$ is the worst solution in $\cal E(\cal G)$, while it is called the approximate price of stability of $\cal G$ ($\epsilon$-PoS$(\cal G)$) when $K$ is the best solution in $\cal E(\cal G)$. For $\epsilon=0$, that is for the set of PNE of $\cal G$, the terms price of anarchy (PoA$(\cal G)$) and price of stability (PoS$(\cal G)$) are used. For $\cal E(\cal G)$ being the set of solutions achieved after one-round walks starting from the empty strategy profile in $\cal G$, $K$ is always defined as the worst solution in $\cal E(\cal G)$ and ${\cal Q}({\cal G},{\cal E},{\cal F})$ is denoted by Apx$_\emptyset^1(\cal G)$.

\

\noindent{\bf Our Contribution.} Our method reveals to be particularly powerful when applied to the class of weighted congestion games. In these games there are $n$ players competing for a set of resources. Each player has an associated weight denoting how much the player congestions a resource when using it, while each resource has an associated latency function modeling the cost incurred by each player when using it. The latency function of each resource only depends on the total weight of the players using it (the overall congestion of the resource) and the cost experienced by each player in a given strategy profile is then defined as the sum of the costs incurred by the player on each of the used resources. These games have a particular appeal since, from one hand, they are general enough to model a variety of situations arising in real life applications and, from the other one, they are structured enough to allow a systematic theoretical study. For example, for the case in which all players have the same weight (the so-called unweighted players), Rosenthal \cite{R73} proved through a potential function argument that PNE are always guaranteed to exist, while general weighted congestion games are guaranteed to possess PNE if and only if the latency functions are either affine or exponential \cite{FKS05,HK10,HKM09,PS06}.

In order to illustrate the versatility and usefulness of our technique, we first consider the well-known and studied case in which the latency functions are affine and $\cal F$ is the sum of the players' payoffs and show how all the known results (as well as some of their generalizations) can be easily reobtained under a unifying approach. They are depicted in Figure \ref{fig1}. Note that, since all the upper bounds for $\epsilon$-PoA and $\epsilon$-PoS are expressed as a function of $\epsilon$, results concerning PNE can be derived by setting $\epsilon=0$. For $\epsilon$-PoA and $\epsilon$-PoS in the unweighted case and for Apx$_\emptyset^1$, we reobtain known upper bounds with significatively shorter and simpler proofs (where, by simple, we mean that only basic notions of calculus are needed in the arguments), while for the generalizations of the $\epsilon$-PoA and the $\epsilon$-PoS in the weighted case, we give the first upper bounds known in the literature.
\begin{figure}\label{fig1}
\begin{center}
\begin{tabular}{|c|c|c|c|}
  \hline
  \ \ Measure\ \ \ & Unweighted & Weighted & \ \ Conditions\ \ \ \\\hline\hline
  $\epsilon$-PoA & $\ \frac{(1+\epsilon)(z^2+3z+1)}{2z-\epsilon}$\ \ \cite{CKS10}\ \ & $\psi^2$ & $\ \epsilon\geq 0$\ \ \\\hline
  $\epsilon$-PoS & $\frac{1+\sqrt{3}}{\epsilon+\sqrt{3}}$\ \ \cite{CKS10} & $\frac{2}{1+\epsilon}$ & $\epsilon\in [0;1]$\\\hline
  Apx$^1_\emptyset$ & $2+\sqrt{5}$\ \ \cite{CMS06} & $\ 4+2\sqrt{3}$\ \ \cite{CMS06}\ \ & $-$\\\hline
\end{tabular}\caption{Upper bounds on the quality of some non-cooperative solutions in weighted and unweighted linear congestion games. Here, $\psi=\frac{1+\epsilon+\sqrt{\epsilon^2+6\epsilon+5}}{2}$ and $z=\left\lfloor\psi\right\rfloor$.}
\end{center}
\end{figure}

After having introduced the technique, we show how it can be used to attack the more challenging case of polynomial latency functions. In such a case, the PoA and $\epsilon$-PoA was already studied and characterized in \cite{ADGMS06} and \cite{CKS10}, respectively, and both papers pose the achievement of upper bounds on the PoS and $\epsilon$-PoS as a major open problem in the area. For unweighted players and quadratic and cubic latency functions, we easily achieve the upper bounds on PoS and Apx$^1_\emptyset$ reported in Figure \ref{fig2}. Extensions to $\epsilon$-PoS and weighted players are left to future work.
\begin{figure}
\begin{center}
\begin{tabular}{|c|c|c|}
  \hline
  \ \ Measure\ \ \ & \ \ Quadratic Latencies \ \ \ & \ \ Cubic Latencies \ \ \ \\\hline\hline
  PoS & $2.362$ & $3.321$ \\\hline
  PoA & $\frac{115}{12}\approx 9.583$\ \ \cite{ADGMS06} & $\frac{1163}{28}\approx 41.535$\ \ \cite{ADGMS06} \\\hline
  Apx$^1_\emptyset$ & $37.5888$ & $\frac{17929}{34}\approx 527.323$ \\\hline
\end{tabular}\caption{Upper bounds on the quality of some non-cooperative solutions in unweighted congestion games with quadratic and cubic latency functions.}\label{fig2}
\end{center}
\end{figure}

What we would like to stress here is that, more than the novelty of the results achieved in this paper, what makes our method significative is its capability of being easily adapted to a variety of particular situations and we are more than sure of the fact that it will prove to be a powerful tool to be exploited in the analysis of the efficiency achieved by different classes of self-emerging solutions in other contexts as well. To this aim, in the Appendix, we show how the method applies also to other social functions, such as the maximum of the players' payoffs (Subsection \ref{secmax}), and to other subclasses of congestion games such as resource allocation games with fair cost sharing (Subsection \ref{secfair}). Note that, in the latter case, as well as in the case of polynomial latency functions, the primal-dual technique proposed in \cite{NR10} cannot be used, since the players' costs are not linear in the variables of the problem.

\

\noindent{\bf Related Works.} The study of the quality of self-emerging solutions in non-cooperative systems initiated with the seminal papers of Koutsoupias and Papadimitriou \cite{KP99} and Anshelevich et al. \cite{ADKTWR04} which introduced, respectively, the notions of price of anarchy and price of stability.

Lot of results have been achieved since then and we recall here only the ones which are closely related to our scenario of application, that is, weighted congestion games with polynomial latency functions.

For affine latency functions and $\cal F$ defined as the sum of the players' payoffs, Christodoulou and Koutsoupias \cite{CK05} show that the PoA is exactly $5/2$ for unweighted players, while Awerbuch et al. \cite{AAE05} show that it rises to exactly $(3+\sqrt{5})/2$ in the weighted case. These bounds keep holding also when considering the price of anarchy of generalizations of PNE such as mixed Nash equilibria and correlated equilibria, as shown by Christodoulou and Koutsoupias in \cite{CK05b}. Similarly, for polynomial latency functions with maximum degree equal to $d$, Aland et al. \cite{ADGMS06} prove that the price of anarchy of all these equilibria is exactly $\Phi_d^{d+1}$ in the weighted case and exactly $\frac{(k+1)^{2d+1}-k^{d+1}(k+2)^d}{(k+1)^{d+1}-(k+2)^d+(k+1)^d-k^{d+1}}$ in the unweighted case, where $\Phi_d$ is the unique non-negative real solution to $(x+1)^d=x^{d+1}$ and $k=\lfloor\Phi_d\rfloor$. These interdependencies have been analyzed by Roughgarden \cite{R09} who proves that unweighted congestion games with latency functions constrained in a given set belong to the class of games for which a so-called ``smoothness argument" applies and that such a smoothness argument directly implies the fact that the prices of anarchy of PNE, mixed Nash equilibria, correlated equilibria and coarse correlated equilibria are always the same when $\cal F$ is the sum of the players' payoffs. Such a result has been extended also to the weighted case by Bhawalkar et al. in \cite{BGR10}.

For the alternative model in which $\cal F$ is defined as the maximum of the players' payoffs, Christodoulou and Koutsoupias \cite{CK05} show a PoA of $\Theta(\sqrt{n})$ in the case of affine latency functions.

For the PoS, only the case of unweighted players, affine latency functions and $\cal F$ defined as the sum of the players' payoffs, has been considered so far. The upper and lower bounds achieved by Caragiannis et al. \cite{CFKKM06} and by Christodoulou and Koutsoupias \cite{CK05b}, respectively, set the PoS to exactly $1+1/\sqrt{3}$. Clearly, being the PoS a best-case measure and being the set of PNE properly contained in the set of all the other equilibrium concepts, again we have a unique bound for PNE and all of its generalizations.


As to $\epsilon$-PNE, in the case of unweighted players, polynomial latency functions and $\cal F$ defined as the sum of the players' payoffs, Christodoulou et al. \cite{CKS10} show that the $\epsilon$-PoA is exactly $\frac{(1+\epsilon)((z+1)^{2d+1}-z^{d+1}(z+2)^d)}{(z+1)^{d+1}-z^{d+1}-(1+\epsilon)((z+2)^d-(z+1)^d)}$, where $z$ is the maximum integer satisfying $\frac{z^{d+1}}{(z+1)^d}<1+\epsilon$, and that, for affine latency functions, the $\epsilon$-PoS is at least $\frac{2(3+\epsilon+\theta\epsilon^2+3\epsilon^3+2\epsilon^4+\theta+\theta\epsilon)}
{6+2\epsilon+5\theta\epsilon+6\epsilon^3+4\epsilon^4-\theta\epsilon^3+2\theta\epsilon^2}$, where $\theta=\sqrt{3\epsilon^3+3+\epsilon+2\epsilon^4}$, and at most $(1+\sqrt{3})/(\epsilon+\sqrt{3})$. 

Finally, for affine latency functions and $\cal F$ defined as the sum of the players' payoffs, Apx$^1_\emptyset$ has been shown to be exactly $2+\sqrt{5}$ in the unweighted case as a consequence of the upper and lower bounds provided, respectively, by Christodoulou et al. \cite{CMS06} and by Bil\`o et al. \cite{BFFM09}, while, for weighted players, Caragiannis et al. \cite{CFKKM06} give a lower bound of $3+2\sqrt{2}$ and Christodoulou et al. \cite{CMS06} give an upper bound of $4+2\sqrt{3}$. For $\cal F$ being the maximum of the players' payoffs, Bil\`o et al. \cite{BFFM09} show that Apx$^1_\emptyset$ is $\Theta(\sqrt[4]{n^3})$ in the unweighted case and affine latency functions.

\

\noindent{\bf Paper Organization.} In next section, we give all the necessary definitions and notation, while in Section \ref{sectech1} we briefly outline the primal-dual method. Then, in Section \ref{secaffine1} we illustrate how it applies to affine latency functions and, finally, in Section \ref{secpoly1} we use it to address the case of quadratic and cubic latency functions. Additional material can be found in the Appendix.

\section{Definitions}
For a given integer $n>0$, we denote as $[n]$ the set $\{1,\ldots,n\}$.

A \emph{weighted congestion game}
$\CG = \left([n], E, (\Sigma_i)_{i \in [n]}, (\ell_e)_{e \in E}, (w_i)_{i
\in [n]}\right)$ is a non-cooperative strategic game in which there is a set $E$ of $m$ \emph{resources} to be shared among the
players in $[n]$. Each player $i$ has an associated weight $w_i\in\RP$ and the special case in which $w_i=1$ for any $i\in [n]$ is called the {\em unweighted case}. The strategy set $\Sigma_i$, for any player $i\in [n]$,
is a subset of resources, i.e., $\Sigma_i \subseteq 2^E$. The set $\Sigma=\times_{i\in [n]}\Sigma_i$ is called the set of {\em strategy profiles} (or {\em solutions}) which can be realized in $\CG$. Given a
strategy profile $S = (s_1, s_2,\ldots, s_n)\in\Sigma$ and a resource $e\in E$,
the sum of the weights of all the players using $e$  in $S$, called the \emph{congestion} of
$e$ in $S$, is denoted by ${\cal L}_e(S)=\sum_{i\in [n]:e\in s_i}w_i$. A
\emph{latency function} $\ell_e : \RP \mapsto \RP$ associates each resource $e\in E$ with a
{\em latency} depending on the congestion of
$e$ in $S$. The \emph{cost} of player $i$ in the strategy profile $S$ is given by
$c_i(S)=\sum_{e \in s_i}\ell_e({\cal L}_e(S))$.
This work is concerned only with \emph{polynomial latency} functions of maximum degree $d$, i.e., the case in which $\ell_e(x) = \sum_{i=0}^d\alpha_{e,i}x^d$ with $\alpha_{e,i}\in \RP$, for any $e\in E$ and $0\leq i\leq d$.

Given a strategy profile $S \in \Sigma$ and a strategy $t\in\Sigma_i$ for player $i$, we denote with $(S_{-i}\diamond t)=(s_1,\ldots,s_{i-1},t,s_{i+1},\ldots,s_n)$ the strategy profile obtained from $S$ when player $i$ changes unilaterally her strategy from $s_i$ to $t$. We say that $S'=(S_{-i}\diamond t)$ is an {\em improving deviation} for player $i$ in $S$ if $c_i(S')<c_i(S)$. Given an $\epsilon\geq 0$, a strategy profile $S$ is an {\em $\epsilon$-approximate pure Nash equilibrium} ($\epsilon$-PNE) if, for any $i\in [n]$ and for any $t\in\Sigma_i$, it holds $c_i(S)\leq (1+\epsilon)c_i(S_{-i}\diamond t)$. For $\epsilon=0$, the set of $\epsilon$-approximate pure Nash equilibria collapses to that of {\em pure Nash equilibria} (PNE), that is, the set of strategy profiles in which no player possesses an improving deviation.

Consider the social function $\CSUM:\Sigma\mapsto\RP$ defined as the sum of the players' costs, that is, $\CSUM(S)=\sum_{i\in [n]}c_i(S)$ and let $S^*$ be the strategy profile minimizing it. Given an $\epsilon\geq 0$ and a weighted congestion game $\CG$, let ${\cal E}(\CG)$ be the set of $\epsilon$-approximate Nash equilibria of $\CG$. The $\epsilon$-approximate price of anarchy of $\CG$ is defined as $\epsilon$-PoA$(\CG)=\max_{S\in {\cal E}(\CG)}\left\{\frac{\CSUM(S)}{\CSUM(S^*)}\right\}$, while the $\epsilon$-approximate price of stability of $\CG$ is defined as $\epsilon$-PoS$(\CG)=\min_{S\in {\cal E}(\CG)}\left\{\frac{\CSUM(S)}{\CSUM(S^*)}\right\}$.

Given a strategy profile $S$ and a player $i\in [n]$, a strategy profile $t^*\in\Sigma_i$ is a {\em best-response} for player $i$ in $S$ if it holds $c_i(S_{-i}\diamond t^*)\leq c_i(S_{-i}\diamond t)$ for any $t\in\Sigma_i$. Let $S^\emptyset$ be the {\em empty strategy profile}, i.e., the profile in which no player has performed any strategic choice yet. A one-round walk starting from the empty strategy profile is an $(n+1)$-tuple of strategy profiles $W=(S^W_0,S^W_1,\ldots,S^W_n)$ such that $S^W_0=S^\emptyset$ and, for any $i\in [n]$, $S^W_i=(S^W_{i-1}\diamond t^*)$, where $t^*$ is a best-response for player $i$ in $S^W_{i-1}$. The profile $S^W_n$ is called the solution achieved after the one-round walk $W$. Clearly, depending on how the players are ordered from 1 to $n$ and on which best-response is selected at step $i$ when more than one best-response is available to player $i$ in $S^W_{i-1}$, different one-round walks can be generated. Let ${\cal W}(\CG)$ denote the set of all possible one-round walks which can be generated in game $\CG$. The approximation ratio of the solutions achieved after a one-round walk starting from the empty strategy profile in $\CG$ is defined as Apx$^1_\emptyset(\CG)=\max_{W\in {\cal W}(\CG)}\left\{\frac{\CSUM(S^W_n)}{\CSUM(S^*)}\right\}$.

\section{The Primal-Dual Technique}\label{sectech1}
Fix a weighted congestion game $\CG$, a social function $\cal F$ and a class of self-emerging solutions $\cal E$. Let $O=(s_1^O,\ldots,s_n^O)$ be the strategy profile optimizing $\cal F$ and $K=(s_1^K,\ldots,s_n^K)\in\cal E(\CG)$ be the worst-case solution in $\cal E(\CG)$ with respect to $\cal F$. For any $e\in E$, we set, for the sake of brevity, $O_e={\cal L}_e(O)$ and $K_e={\cal L}_e(K)$.

Since $O$ and $K$ are fixed, we can maximize the inefficiency yielded by the pair $(K,O)$ by suitably choosing the coefficients $\alpha_{e,i}$, for each $e\in E$ and $0\leq i\leq d$, so that ${\cal F}(K)$ is maximized, ${\cal F}(O)$ is normalized to one and $K$ meets all the constraints defining the set $\cal E(\cal G)$. For the sets $\cal E$ and social functions $\cal F$ considered in this paper, this task can be easily achieved by creating a suitable linear program $LP(K,O)$.

By providing a feasible solution for the dual program $DLP(K,O)$, we can obtain an upper bound on the optimal solution of $LP(K,O)$. Our task is to uncover, among all possibilities, the pair $(K^*,O^*)$ yielding the highest possible optimal solution for $LP(K,O)$. To this aim, the study of the dual formulation plays a crucial role: if we are able to detect the nature of the ``worst-case" dual constraints, then we can easily figure out the form of the pair $(K^*,O^*)$ maximizing the inefficiency of the class of solutions $\cal E$. Clearly, by the complementary slackness conditions, if we find the optimal dual solution, then we can quite systematically construct the matching primal instance by choosing a suitable set of players and resources so as to implement all the tight dual constraints. This task is much more complicated to be achieved in the weighted case, because, once established the values of the congestions $K_e^*$ and $O_e^*$ for any $e\in E$, there are still infinite many ways to split them among the players using resource $e$ in both $K$ and $O$. For such a reason, discovering matching lower bounding instances for the weighted case reveals to be much harder than for the unweighted one.

\section{Application to Affine Latency Functions}\label{secaffine1}
In order to easily illustrate our primal-dual technique, in this section we consider the well-known and studied case of affine latency functions and social function $\CSUM$ and show how the results for $\epsilon$-PoA, $\epsilon$-PoS and Apx$^1_\emptyset$ already known in the literature can be reobtained in a unified manner for both weighted and unweighted players.

In order to reduce the number of variables we need to consider in our linear programs, we make use of the following simplificative arguments.

We say that the game $\CG'=([n],E',\Sigma',\ell',w)$ is equivalent to the game $\CG=([n],E,\Sigma,\ell,w)$ if there exists a one-to-one mapping $\varphi_i:\Sigma_i\mapsto\Sigma'_i$ for any $i\in [n]$ such that for any strategy profile $S=(s_1,\ldots,s_n)\in\Sigma$, it holds $c_i(S)=c_i(\varphi_1(s_1),\ldots,\varphi_n(s_n))$ for any $i\in [n]$.

\begin{lemma}\label{lemma1}
For each weighted congestion game with affine latency functions $\CG=([n],E,\Sigma,\ell,w)$ there always exists an equivalent weighted congestion game with affine latency functions $\CG'=([n],E',\Sigma',\ell',w)$ such that $\ell'_e(x) = \alpha_{e,1} x:=\alpha_e x$ for any $e\in E'$.
\end{lemma}

Because of this lemma, throughout this section, we are allowed to restrict our attention only to games with latency functions of the form $\ell_e(x) = \alpha_e x$, for any $e\in E$. In such a setting, we can rewrite the social value of a strategy profile as $\CSUM(S)=\sum_{e\in E}(\alpha_e {\cal L}_e(S)^2)$.

\subsection{Bounding the Approximate Price of Anarchy}
By definition, we have that $K$ is an $\epsilon$-PNE if, for any $i\in [n]$, it holds $$c_i(K)=\sum_{e\in s_i^K}(\alpha_e K_e)\leq (1+\epsilon)c_i(K_{-i}\diamond s_i^O)\leq(1+\epsilon)\sum_{e\in s_i^O}(\alpha_e(K_e+w_i)).$$ Thus, the primal formulation $LP(K,O)$ assumes the following form.
\begin{displaymath}
\begin{array}{rlr}
maximize & \displaystyle\sum_{e\in E}\left(\alpha_e K_e^2\right)\\
subject\ to\\
\displaystyle\sum_{e\in s_i^K}\left(\alpha_e K_e\right)-(1+\epsilon)\sum_{e\in s_i^O}\left(\alpha_e (K_e+w_i)\right) & \leq 0, & \forall i\in [n]\\
\displaystyle\sum_{e\in E}\left(\alpha_e O_e^2\right) & = 1,\\
\alpha_e & \geq 0, & \forall e\in E
\end{array}
\end{displaymath}
The dual program $DLP(K,O)$ is
\begin{displaymath}
\begin{array}{rlr}
minimize & \gamma\\
subject\ to\\
\displaystyle\sum_{i:e\in s_i^K}\left(y_i K_e\right)-(1+\epsilon)\sum_{i:e\in s_i^O}\left(y_i (K_e+w_i)\right)+\gamma O_e^2 & \geq K_e^2, & \forall e\in E\\
y_i & \geq 0, & \forall i\in [n]
\end{array}
\end{displaymath}
Let $\psi=\frac{1+\epsilon+\sqrt{\epsilon^2+6\epsilon+5}}{2}$ and $z=\lfloor\psi\rfloor$. For unweighted players we reobtain the upper bound proved in \cite{CKS10} with a much simpler and shorter proof.

\begin{theorem}\label{poa-un}
For any $\epsilon\geq 0$ and $\cal G$ with unweighted players, it holds $\epsilon$-PoA$({\cal G})\leq\frac{(1+\epsilon)(z^2+3z+1)}{2z-\epsilon}$.
\end{theorem}
\begin{proof}
In such a case, since $w_i=1$ for each $i\in [n]$, by choosing $y_i=\frac{2z+1}{2z-\epsilon}$ for any $i\in [n]$ and $\gamma=\frac{(1+\epsilon)(z^2+3z+1)}{2z-\epsilon}$, the dual inequalities become of the form $$\frac{2z+1}{2z-\epsilon}\left(K_e^2-(1+\epsilon)(K_e+1)O_e\right)+\frac{(1+\epsilon)(z^2+3z+1)}{2z-\epsilon} O_e^2\geq K_e^2$$ which is equivalent to
\begin{equation}
K_e^2-(2z+1)(K_e O_e+O_e)+(z^2+3z+1)O_e^2\geq 0.\label{eq1}
\end{equation}
Easy calculations show that this is always verified for any pair of non-negative integers $(K_e,O_e)$. Note that the definition of $z$ guarantees that $2z-\epsilon\geq 0$, so that the proposed dual solution is a feasible one.\qed
\end{proof}

When $\epsilon=0$, we reobtain the well-known price of anarchy of $5/2$ which holds for PNE. We illustrate how the dual formulation can be also used to discover a matching lower bounding instance. From the analysis of the dual constraints (\ref{eq1}), it is easy to see that they get tight only for pairs $(K_e,O_e)$ of the form $(1,1)$ and $(2,1)$. Thus, if the $5/2$ upper bound is tight, the complementary slackness conditions assure us that in the matching lower bounding instance only resources implementing the pairs $(1,1)$ and $(2,1)$ are needed. This can be easily achieved through a game using 3 players and 3 resources and defined as follows: $\Sigma_1=\{\{e_1,e_2\},\{e_3\}\}$, $\Sigma_2=\{\{e_1\},\{e_2,e_3\}\}$, $\Sigma_3=\{\{e_2\},\{e_3\}\}$ and $\alpha_1=5$, $\alpha_2=2$, $\alpha_3=3$. For such an instance, we have $K=(\{e_1,e_2\},\{e_2,e_3\},\{e_3\})$, $O=(\{e_3\},\{e_1\},\{e_2\})$ with a price of anarchy of $5/2$. Clearly, this instance can be extended to any number of players $n>3$ by adding a fourth resource $e_4$ with $\alpha_4=0$ and setting $\Sigma_i=\{e_4\}$ for any $i\in [n]$ with $i\geq 4$. Note that these are minimal lower bounding instances (the previous known lower bounding instances presented in \cite{CK05} used $2n$ resources for any $n\geq 3$).

More generally, for $\epsilon$-PNE, the dual constraints (\ref{eq1}) get tight only for pairs of the form $(z,1)$ and $(z+1,1)$. Thus, in order to obtain a matching lower bounding instance, we only need to implement this family of dual constraints, that is, we need an instance with at least $z+2$ players and a set of resources such that $O_e=1$ and $K_e\in\{z,z+1\}$ for any $e\in E$. In fact, the matching lower bounding instances given in \cite{CKS10} use $z+2$ players and $2z+4$ resources, half of which has $K_e=z$ and the other half has $K_e=z+1$. It is easy to see that these instances are not minimal. In fact, they produce a dual program with $z+3$ variables and $2z+4$ constraints, where only $z+3$ constraints are sufficient to exactly characterize the optimal dual solution. Unfortunately, this set of constraints changes as a function of $\epsilon$ and so it is not easy to achieve a general scheme of minimal lower bounding instances.

For the weighted case, we can prove the following upper bound.

\begin{theorem}\label{poa-w}
For any $\epsilon\geq 0$ and $\cal G$ with weighted players, it holds $\epsilon$-PoA$({\cal G})\leq\psi^2$.
\end{theorem}
\begin{proof}
In such a case, by choosing $y_i=1+\frac{\sqrt{1+\epsilon}}{\sqrt{5+\epsilon}}$ for any $i\in [n]$ and $\gamma=\psi^2$, the dual inequalities become of the form $$\left(1+\frac{\sqrt{1+\epsilon}}{\sqrt{5+\epsilon}}\right)\left(K_e^2-(1+\epsilon)\left(K_e O_e+O_e\right)\right)+\psi^2 O_e^2\geq K_e^2$$ which is equivalent to
\begin{equation}
\frac{\sqrt{1+\epsilon}}{\sqrt{5+\epsilon}}K_e^2-\left(1+\frac{\sqrt{1+\epsilon}}{\sqrt{5+\epsilon}}\right)(1+\epsilon)(K_e O_e+O_e)+\psi^2 O_e^2\geq 0.\label{eq2}
\end{equation}
Easy calculations show that this is always verified for any pair of non-negative reals $(K_e,O_e)$.\qed
\end{proof}

Note that, when $\psi=z$, it holds $\psi^2=\frac{(1+\epsilon)(z^2+3z+1)}{2z-\epsilon}$. Hence, the $\epsilon$-PoA in the weighted and unweighted cases coincide for all $\epsilon\geq 0$ such that $\frac{1+\epsilon+\sqrt{\epsilon^2+6\epsilon+5}}{2}$ is a natural number. The dual constraints (\ref{eq2}) get tight only for pairs $(K_e,O_e)$ such that $K_e=\psi O_e$.

For PNE, by setting $\epsilon=0$ we reobtain the well-known bound $(3+\sqrt{5})/2$. We can show that this is tight by considering an instance having 3 players and 3 resources with $w_1=1$, $w_2=w_3=(1+\sqrt{5})/2$, $\Sigma_1=\{\{e_1\},\{e_2,e_3\}\}$, $\Sigma_2=\{\{e_2\},\{e_1,e_3\}\}$, $\Sigma_3=\{\{e_3\},\{e_2\}\}$, $\alpha_1=2$, $\alpha_2=\sqrt{5}-1$ and $\alpha_3=3-\sqrt{5}$. We have $K=(\{e_2,e_3\},\{e_1,e_3\},\{e_2\})$, $O=(\{e_1\},\{e_2\},\{e_3\})$ and the price of anarchy is equal to $(3+\sqrt{5})/2$. Again, we have identified a minimal lower bounding instance which is slightly simpler than the previous one given in \cite{AAE05} which used $4$ players and $3$ resources.

For general $\epsilon$-PNE, we are able to provide tight lower bounds only for a subset, although having infinite cardinality, of values of $\epsilon$. Let $t$ and $y$ be two positive integers such that $1\leq y\leq t+1$. We set $\epsilon(t,y)=\frac{(t-1)\sqrt{t^2+4y}+2y+t^2-t-2}{\sqrt{t^2+4y}+t+2}$, which is always non-negative since $y\geq 1$ and yields $\psi=\frac{t+\sqrt{t^2+4y}}{2}>t$. We create an instance with $t+2$ players and $2(t+1)$ resources, where $w_i=1$ for any $i\in [t+1]$ and $w_{t+2}=\psi-t$. The first $t+1$ resources $(e_j)_{j\in [t+1]}$ have latency $\ell(x)=x$, while the last $t+1$ resources $(e'_j)_{j\in [t+1]}$ have latency $\ell(x)=x/y$. The set of strategies for each player $i\in [t+1]$ is $\Sigma_i=\{\{e_i\},\bigcup_{j\in [t]}\{e_{i+j}\}\cup\bigcup_{j\in [y]}\{e'_{i+j}\}\}$, with the sum of the indexes taken circularly, while $\Sigma_{t+2}=\{\bigcup_{j\in [t+1]}\{e'_{j}\},\bigcup_{j\in [t+1]}\{e_j\}\}$. The first strategy of each player is the optimal one, while the second strategy is the one played at the $\epsilon(t,y)$-PNE. Note that for any $e$ we have $K_e=\psi O_e$, thus implying an $\epsilon$-PoA equal to $\psi^2$. It is not difficult to show that $K$ is an $\epsilon(t,y)$-PNE by exploiting the equality $2\psi=t+\sqrt{t^2+4y}$.

Deriving a tight lower bound for any possible value of $\epsilon$ remains an interesting open problem.

\subsection{Bounding the Approximate Price of Stability}\label{pos}
Recall that, since the $\epsilon$-PoS is a best-case measure, the primal-dual approach guarantees a tight analysis only if we are able to exactly characterize the polyhedron defining the set of the best $\epsilon$-PNE. It is not known how to do this at the moment, thus all the approaches used so far in the literature approximate the best $\epsilon$-PNE with an $\epsilon$-PNE minimizing a certain potential function. In \cite{CKS10}, it is shown that, for unweighted players, any strategy profile $S$ which is a local minimum of the function $\Phi_\epsilon(S)=\sum_{e\in E}\left(\alpha_e\left({\cal L}_e(S)^2+\frac{1-\epsilon}{1+\epsilon}{\cal L}_e(S)\right)+\frac{2\beta_e}{1+\epsilon}{\cal L}_e(S)\right)$, called $\epsilon$-approximate potential, is an $\epsilon$-PNE. Thus, it is possible to get an upper bound on the $\epsilon$-PoS by measuring the $\epsilon$-PoA of the global minimum of $\Phi_\epsilon$.

We now illustrate our approach which yields the same $\frac{1+\sqrt{3}}{\epsilon+\sqrt{3}}$ upper bound achieved in \cite{CKS10}. First of all, we can use the inequality $\Phi_\epsilon(K)\leq\Phi_\epsilon(O)$ which results in the constraint
$$\displaystyle\sum_{e\in E}\left(\alpha_e \left(K^2_e+\frac{1-\epsilon}{1+\epsilon}K_e-O^2_e-\frac{1-\epsilon}{1+\epsilon}O_e\right)\right)\leq 0.$$
Then, since we assume that $K$ is the global minimum of $\Phi_\epsilon$, we also have $\sum_{i\in [n]}\left(\Phi_\epsilon(K)-\Phi_\epsilon(K_{-i}\diamond s_i^O)\right)\leq 0$ which results in the constraint
$$\displaystyle\sum_{e\in E}\left(\alpha_e \left(K^2_e-\frac{\epsilon}{1+\epsilon}K_e-K_e O_e-\frac{1}{1+\epsilon}O_e\right)\right)\leq 0.$$
Thanks to this, the dual formulation becomes the following one.
\begin{displaymath}
\begin{array}{rlr}
minimize & \gamma\\
subject\ to\\
K_e^2(x+y)+\frac{K_e}{1+\epsilon}(y(1-\epsilon)-z\epsilon)\\
-\left(y O_e^2+z K_e O_e+\frac{O_e}{1+\epsilon}(y(1-\epsilon)+z)\right)+\gamma O_e^2 & \geq K_e^2, & \forall e\in E\\
y,z & \geq 0
\end{array}
\end{displaymath}
Thus, for unweighted players, we obtain the following result for any $\epsilon\in [0;1]$ (this is the only interesting case, since \cite{CKS10} shows that, for any $\epsilon\geq 1$, $\epsilon$-PoS$({\cal G})=1$ for any $\cal G$).

\begin{theorem}\label{pos-un}
For any $\epsilon\in [0;1]$ and $\cal G$ with unweighted players, it holds $\epsilon$-PoS$({\cal G})\leq\frac{1+\sqrt{3}}{\epsilon+\sqrt{3}}$.
\end{theorem}
\begin{proof}
By choosing $y=\frac{2\epsilon+\sqrt{3}(1+\epsilon)}{2(\epsilon+\sqrt{3})}$, $z=\frac{1-\epsilon}{\epsilon+\sqrt{3}}$ and $\gamma=\frac{1+\sqrt{3}}{\epsilon+\sqrt{3}}$, the dual inequalities become
$$(\epsilon-1)((\sqrt{3}-2)K_e^2+(2O_e-\sqrt{3})K_e+(2+\sqrt{3})(O_e-O_e^2))\geq 0.$$
Easy calculations show that this is always verified for any pair of non-negative integers $(K_e,O_e)$.\qed
\end{proof}

Here, the only tight dual constraints are those of the form $(0,1)$ and $(1,1)$ which are clearly insufficient to achieve an $\epsilon$-PoS greater than $1$, since $K_e\leq O_e$ for any $e\in E$. What is going on here? The answer is that the lower bound on the $\epsilon$-PoS can be achieved only asymptotically, that is, when $n$ tends to infinity. Thus, we must also check what happens when both $K_e$ and $O_e$ goes to infinity and their ratio remains constant. We obtain that the dual constraints are asymptotically tight for pairs of the form $(K_e,O_e)$ such that $K_e=(2+\sqrt{3})O_e$ and $O_e$ goes to infinity. The lower bounding instances proposed in \cite{CKS10} have $n_1$ resources of type $(0,1)$, $n_1(n_1-1)$ resources of type $(1,1)$ and one resource of type $\left(n_1,\frac{\sqrt{2\epsilon^4+3\epsilon^3+\epsilon+3}+2\epsilon^2+2\epsilon-1}{\sqrt{2\epsilon^4+3\epsilon^3+\epsilon+3}+\epsilon^2+\epsilon+1}n_1\right)$ and $n_1$ going to infinity. Thus, such lower bounding instances possess all the combinatorics needed to implement the worst-case dual constraints, but still there is a remarkable gap between upper and lower bounds. Hence, the intuition should suggest us that that the upper bound is not tight and additional constraints should be used in the primal formulation so as to better characterize the polyhedron defining the best $\epsilon$-PNE. Note that the inequalities stating the $K$ is an $\epsilon$-PNE is of no use here since they are dominated by the inequality $\sum_{i\in [n]}\left(\Phi_\epsilon(K)-\Phi_\epsilon(K_{-i}\diamond s_i^O)\right)\leq 0$.

\

In order to deal with the weighted case, it is possible to rephrase the approach of \cite{CKS10} to turn the potential given in \cite{FKS05} for weighted linear congestion games into an $\epsilon$-potential function for this class of games so as to use the same approach as in the unweighted case. Anyway, no particularly significative upper bounds can be achieved in this case as shown in the following theorem (details can be found in subsection \ref{PoSW} of the Appendix).

\begin{theorem}\label{posw}
For any $\epsilon\in [0;1]$ and $\cal G$ with weighted players, it holds $\epsilon$-PoS$({\cal G})\leq\frac{2}{1+\epsilon}$.
\end{theorem}

In this case, no specific lower bounds are known.

\subsection{Bounding the Approximation Ratio of One-Round Walks}
For a one-round walk $W$, we set $K=S_n^W$. Define $K_e(i)$ as the sum of the weights of the players using resource $e$ in $K$ before player $i$ performs her choice. $LP(K,O)$ in this case has the following form.

\begin{displaymath}
\begin{array}{rlr}
maximize & \displaystyle\sum_{e\in E}\left(\alpha_e K_e^2\right)\\
subject\ to\\
\displaystyle\sum_{e\in s_i^K}\left(\alpha_e (K_e(i)+w_i)\right)-\sum_{e\in s_i^O}\left(\alpha_e (K_e(i)+w_i)\right) & \leq 0, & \forall i\in [n]\\
\displaystyle\sum_{e\in E}\left(\alpha_e O_e^2\right) & = 1\\
\alpha_e & \geq 0, & \forall e\in E
\end{array}
\end{displaymath}
$DLP(K,O)$ is as follows.
\begin{displaymath}
\begin{array}{rlr}
minimize & \gamma\\
subject\ to\\
\displaystyle\sum_{i:e\in s_i^K}\left(y_i (K_e(i)+w_i)\right)-\sum_{i:e\in s_i^O}\left(y_i (K_e(i)+w_i)\right)+\gamma O_e^2 & \geq K_e^2, & \forall e\in E\\
y_i & \geq 0, & \forall i\in [n]
\end{array}
\end{displaymath}

For both unweighted and weighted players we easily reobtain the upper bounds on Apx$^1_\emptyset$ given in \cite{CMS06}.

\begin{theorem}\label{apx-un}
For any $\cal G$ with unweighted players, it holds Apx$^1_\emptyset({\cal G})\leq 2+\sqrt{5}$.
\end{theorem}
\begin{proof}
Observe that the worst-case dual constraints occur when each player $i$ using resource $e$ in $O$ enters the game after all players using $e$ in $K$ have entered the game yet. In such a case, by choosing $y_i=1+\sqrt{5}$ for any $i\in [n]$ and $\gamma=2+\sqrt{5}$, the worst-case dual inequalities become $$\left(1+\sqrt{5}\right)\left(\frac{K_e(K_e+1)}{2}-(K_e+1)O_e\right)+\left(2+\sqrt{5}\right)O_e^2\geq K_e^2$$ which is equivalent to $$\left(\frac{\sqrt{5}-1}{2}\right)K_e^2+\left(\frac{1+\sqrt{5}}{2}\right)K_e-(1+\sqrt{5})K_e O_e-(1+\sqrt{5})O_e+(2+\sqrt{5})O_e^2\geq 0.$$ Easy calculations show that this is always verified for any pair of non-negative integers $(K_e,O_e)$.\qed
\end{proof}

The dual constraints get tight for pairs of the form $(1,1)$, while the asymptotical dual constraints get tight for pairs of the form $(\frac{3+\sqrt{5}}{2}O_e,O_e)$. These pairs exactly characterize the structure of the lower bounding instance derived in \cite{BFFM09}.

For weighted players, a slightly more involved analysis is needed.

\begin{theorem}\label{apx-w}
For any $\cal G$ with weighted players, it holds Apx$^1_\emptyset({\cal G})\leq 4+2\sqrt{3}$.
\end{theorem}
\begin{proof}
Again, observe that the worst-case dual constraints occur when each player $i$ using resource $e$ in $O$ enters the game after all players using $e$ in $K$ have entered the game yet. In such a case, by choosing $y_i=\left(2+\frac{2}{\sqrt{3}}\right)w_i$ for any $i\in [n]$ and $\gamma=4+2\sqrt{3}$, the worst-case dual inequalities become $$\left(2+\frac{2}{\sqrt{3}}\right)\left(\sum_{i\leq j:e\in s_i^K\cap s_j^K}(w_i w_j)-\sum_{i:e\in s_i^O}\left(w_i(K_e+w_i)\right)\right)+\left(4+2\sqrt{3}\right)O_e^2\geq K_e^2$$ which is true if it holds $$\frac{1}{\sqrt{3}}K_e^2-\left(2+\frac{2}{\sqrt{3}}\right)K_e O_e+\left(2+\frac{4}{\sqrt{3}}\right)O_e^2\geq 0.$$ Easy calculations show that this is always verified for any pair of non-negative reals $(K_e,O_e)$.\qed
\end{proof}

Note that this is the only case among the ones considered so far for which the dual variables are not player independent. The worst case dual constraints occur when all players using resource $e$ in the walk have weight $1$, while only one player uses $e$ in the optimal solution. Moreover, the asymptotical dual constraints get tight for pairs of the form $((1+\sqrt{3})O_e,O_e)$. In this case, the best known lower bound, equal to $3+2\sqrt{2}$, has been given in \cite{CFKKM06}.

\section{Quadratic and Cubic Latency Functions}\label{secpoly1}
In this section, we show how to use the primal-dual method to bound PoS and Apx$^1_\emptyset$ in the case of polynomial latency functions of maximum degree $d$ and unweighted players. We only consider the case $d\leq 3$, that is, quadratic and cubic latency functions. It is not difficult to extend the approach to any particular value of $d$, but it is quite hard to obtain a general result as a function of $d$ because we do not have simple closed formulas expressing some of the summations we need in our analysis for any value of $d$. We also leave the extension to $\epsilon$-PNE and weighted players for future works. We restrict to the cases in which the latency functions are of the form $\ell_e(x)=\alpha_e x^d$, since it is possible to show that this can be supposed without loss of generality.

\subsection{Bounding the Price of Stability}
For $d=2$, the potential function is $\Phi(S)=\frac 1 6\sum_{e\in E}{\cal L}_e(S)({\cal L}_e(S)+1)(2{\cal L}_e(S)+1)$. Thus, the constraint $\Phi(K)\leq\Phi(O)$ becomes $\sum_{e\in E}\left(K_e(K_e+1)(2K_e+1)-O_e(O_e+1)(2O_e+1)\right)\leq 0.$ Similarly, the constraint $\sum_{i\in [n]}\left(\Phi(K)-\Phi(K_{-i}\diamond s_i^O)\right)\leq 0$ becomes $\sum_{e\in E}\left(K_e^3-O_e(K_e+1)^2\right)\leq 0.$
Hence, $DLP(K,O)$ is defined as follows.

\begin{displaymath}
\begin{array}{rlr}
minimize & \gamma\\
subject\ to\\
\left(y(K_e(K_e+1)(2K_e+1)-O_e(O_e+1)(2O_e+1))\right)\\
+\left(z(K_e^3-O_e(K_e+1)^2)\right)+\gamma O_e^3 & \geq K_e^3, & \forall e\in E\\
y,z & \geq 0
\end{array}
\end{displaymath}

\begin{theorem}
For any $\cal G$ with quadratic latency functions and unweighted players, it holds PoS$({\cal G})\leq 2.362$.
\end{theorem}
\begin{proof}
The claim follows by setting $y=0.318$, $z=0.453$ and $\gamma=2.362$.\qed
\end{proof}

For $d=3$, the potential function is $\Phi(S)=\frac 1 4\sum_{e\in E}\left({\cal L}_e(S)({\cal L}_e(S)+1)\right)^2$. Thus, the constraint $\Phi(K)\leq\Phi(O)$ becomes $\sum_{e\in E}\left(\left(K_e(K_e+1)\right)^2-\left(O_e(O_e+1)\right)^2\right)\leq 0.$ Similarly, the constraint $\sum_{i\in [n]}\left(\Phi(K)-\Phi(K_{-i}\diamond s_i^O)\right)\leq 0$ becomes $\sum_{e\in E}\left(K_e^4-O_e(K_e+1)^3\right)\leq 0.$

Hence, $DLP(K,O)$ is defined as follows.
\begin{displaymath}
\begin{array}{rlr}
minimize & \gamma\\
subject\ to\\
\left(y(K_e^2(K_e+1)^2-O_e^2(O_e+1)^2)\right)+\left(z(K_e^4-O_e(K_e+1)^3)\right)+\gamma O_e^4 & \geq K_e^4, & \forall e\in E\\
y,z & \geq 0
\end{array}
\end{displaymath}

\begin{theorem}
For any $\cal G$ with cubic latency functions and unweighted players, it holds PoS$({\cal G})\leq 3.322$.
\end{theorem}
\begin{proof}
The claim follows by setting $y=0.747$, $z=0.331$ and $\gamma=3.322$.\qed
\end{proof}

It is not difficult to extend the instances given in \cite{CKS10} so as to obtain lower bounds of $2.1859$ and $2.7558$, respectively (see Subsection \ref{seclb} in the Appendix).

\subsection{Bounding the Approximation Ratio of One-Round Walks}
For $d=2$, $DLP(K,O)$ is defined as follows.
\begin{displaymath}
\begin{array}{rlr}
minimize & \gamma\\
subject\ to\\
\displaystyle\sum_{i:e\in s_i^K}\left(y_i (K_e(i)+1)^2\right)-\sum_{i:e\in s_i^O}\left(y_i (K_e(i)+1)^2\right)+\gamma O_e^3 & \geq K_e^3, & \forall e\in E\\
y_i & \geq 0, & \forall i\in [n]
\end{array}
\end{displaymath}

\begin{theorem}\label{apx-2}
For any $\cal G$ with quadratic latency functions and unweighted players, it holds Apx$^1_\emptyset({\cal G})\leq 37.5888$.
\end{theorem}
\begin{proof}
As usual, the worst-case dual constraints occur when each player $i$ using resource $e$ in $O$ enters the game after all players using $e$ in $K$ have entered the game yet. The claim follows by choosing $y_i=5.2944$ for any $i\in [n]$ and $\gamma=37.5888$.\qed
\end{proof}

For $d=3$, $DLP(K,O)$ is defined as follows.
\begin{displaymath}
\begin{array}{rlr}
minimize & \gamma\\
subject\ to\\
\displaystyle\sum_{i:e\in s_i^K}\left(y_i (K_e(i)+1)^3\right)-\sum_{i:e\in s_i^O}\left(y_i (K_e(i)+1)^3\right)+\gamma O_e^4 & \geq K_e^4, & \forall e\in E\\
y_i & \geq 0, & \forall i\in [n]
\end{array}
\end{displaymath}

\begin{theorem}\label{apx-3}
For any $\cal G$ with cubic latency functions and unweighted players, it holds Apx$^1_\emptyset({\cal G})\leq\frac{17929}{34}\approx 527.323$.
\end{theorem}
\begin{proof}
As usual, the worst-case dual constraints occur when each player $i$ using resource $e$ in $O$ enters the game after all players using $e$ in $K$ have entered the game yet. The claim follows by choosing $y_i=\frac{369}{34}$ for any $i\in [n]$ and $\gamma=\frac{17929}{34}$.\qed
\end{proof}

\

\noindent{\bf Acknowledgments:} Some of the ideas this work relies upon originated after several seminal discussions with Michele Flammini, Angelo Fanelli and Luca Moscardelli to whom the author is greatly indebted.

\newpage
\section{Appendix}

\subsection{Proof of Lemma \ref{lemma1}}
\begin{proof}
Consider the weighted congestion game $\CG=([n],E,\Sigma,\ell,w)$ with latency functions $\ell_e(x) = \alpha_e x + \beta_e$ for any $e\in E$. For each $\widetilde{e}\in E$ such that $\beta_{\widetilde{e}} > 0$, let $N_{\widetilde{e}}$ be the set of players who can choose $\widetilde{e}$, that is, $N_{\widetilde{e}}=\{i\in [n]:\exists s\in \Sigma_i: \widetilde{e}\in s\}$. The set of resources $E'$ is obtained by replicating all the resources in $E$ and adding a new resource $e_{\widetilde{e}}^i$ for any $\widetilde{e}\in E$ and any $i\in N_{\widetilde{e}}$, that is, $E'=E\cup\bigcup_{\widetilde{e}\in E,i\in N_{\widetilde{e}}}\{e_{\widetilde{e}}^i\}$. The latency functions are defined as $\ell'_e(x)=\alpha_e x$ for any $e\in E'\cap E$ and $\ell'_{e_{\widetilde{e}}^i}(x)=\frac{\beta_{\widetilde{e}}}{w_i} x$ for any $\widetilde{e}\in E$ and any $i\in N_{\widetilde{e}}$. Finally, for any $i\in [n]$, the mapping $\varphi_i$ is defined as follows: $\varphi_i(s)=s\cup\bigcup_{\widetilde{e}\in s}\{e^i_{\widetilde{e}}\}$. It is not difficult to see that for any $S=(s_1,\ldots,s_n)\in\Sigma$ and for any $i\in [n]$, it holds $c_i(S)=c_i(\varphi_1(s_1),\ldots,\varphi_n(s_n))$.\qed
\end{proof}

\subsection{Proof of Theorem \ref{posw}}\label{PoSW}

We define the following $\epsilon$-potential function.
$$\Phi_\epsilon(S)=\frac 1 2 \sum_{e\in E}\left(\alpha_e {\cal L}_e(S)^2\right)+\frac 1 2 \frac{1-\epsilon}{1+\epsilon}\sum_{e\in E}\sum_{i:e\in s_i}\left(\alpha_e w_i^2\right).$$

\begin{lemma}
Any profile which is a local minimum of $\Phi_\epsilon$ is an $\epsilon$-PNE.
\end{lemma}
\begin{proof}
Consider a profile $S=(s_1,\ldots,s_n)$. We want to compute the change in the $\epsilon$-potential function when player $i$ changes her strategy from $s_i$ to $t$. The resulting profile $(S_{-i}\diamond t)$ has
\begin{displaymath}
{\cal L}_e(S_{-i}\diamond t) =\left\{
\begin{array}{lc}
{\cal L}_e(S)-w_i, & \ \ e\in s_i\setminus t,\\
{\cal L}_e(S)+w_i, & \ \ e\in t\setminus s_i,\\
{\cal L}_e(S), & \ \ otherwise.
\end{array}\right.
\end{displaymath}
From this we can compute the difference
$$\Phi_\epsilon(S_{-i}\diamond t)-\Phi_\epsilon(S)=\sum_{e\in t\setminus s_i}\left(\alpha_e\left(w_i {\cal L}_e(S)+\frac{w_i^2}{1+\epsilon}\right)\right)-\sum_{e\in s_i\setminus t}\left(\alpha_e\left(w_i {\cal L}_e(S)-\frac{\epsilon w_i^2}{1+\epsilon}\right)\right).$$
We can rewrite this as
$$\Phi_\epsilon(S_{-i}\diamond t)-\Phi_\epsilon(S)=$$ $$\sum_{e\in t}\left(\alpha_e\left(w_i {\cal L}_e(S)+\frac{w_i^2}{1+\epsilon}\right)\right)-\sum_{e\in t \cap s_i}\left(\alpha_e w_i^2\right)-\sum_{e\in s_i}\left(\alpha_e\left(w_i {\cal L}_e(S)-\frac{\epsilon w_i^2}{1+\epsilon}\right)\right).$$
Suppose now that $S$ is a local minimum of $\Phi_\epsilon$ which implies $\Phi_\epsilon(S)\leq\Phi_\epsilon(S_{-i}\diamond t)$ for any $i\in [n]$ and $t\in\Sigma_i$. The cost of player $i$ before the change is $c_i(S)=\sum_{e\in s_i}\left(\alpha_e {\cal L}_e(S)\right)$ and after the change is $c_i(S_{-i}\diamond t)=\sum_{e\in t}\left(\alpha_e {\cal L}_e(S_{-i}\diamond t)\right)$. We show that $S$ is an $\epsilon$-PNE, that is, $c_i(S)\leq (1+\epsilon)c_i(S_{-i}\diamond t)$.

By exploiting the two different parts defining the $\epsilon$ potential we obtain
$$c_i(S)=\sum_{e\in s_i}\left(\alpha_e {\cal L}_e(S)\right)\leq\sum_{e\in s_i}\left(\alpha_e(1+\epsilon) \left({\cal L}_e(S)-\frac{\epsilon w_i}{1+\epsilon}\right)\right)$$
which holds because ${\cal L}_e(S)\geq w_i$ when $e\in s_i$, and
$$c_i(S_{-i}\diamond t)=\sum_{e\in t}\left(\alpha_e ({\cal L}_e(S_{-i}\diamond t)+w_i)\right)-\sum_{e\in t \cap s_i}\left(\alpha_e w_i\right)$$ $$\geq\sum_{e\in t}\left(\alpha_e \left({\cal L}_e(S)+\frac{w_i}{1+\epsilon}\right)\right)-\sum_{e\in t \cap s_i}\left(\alpha_e w_i\right)$$
which holds for any $\epsilon\geq 0$.

It follows immediately that $c_i(S)\leq (1+\epsilon)c_i(S_{-i}\diamond t)$, thus $S$ is an $\epsilon$-PNE.\qed
\end{proof}

Using the constraint $\Phi_\epsilon(K)-\Phi_\epsilon(O)\leq 0$ in our formulation, we can easily prove Theorem \ref{posw}.

\begin{proof}
In such a case, by choosing $y_i=0$ for any $i\in [n]$, $z=1$ and $y_{n+1}=\frac{2}{1+\epsilon}$, the dual inequalities become $$\frac{1-\epsilon}{1+\epsilon}K_e\geq 0$$ which is always verified for any pair of non-negative reals $(K_e,O_e)$.\qed
\end{proof}

\subsection{Lower Bounds on the PoS of Quadratic and Cubic Latency Functions}\label{seclb}
\begin{theorem}
For any $\delta > 0$, there exist an unweighted congestion game with quadratic latency functions
${\cal G}_1$ and an unweighted congestion game with cubic latency functions ${\cal G}_2$ such that $PoS({\cal G}_1)\geq 2.1859-\delta$ and
$PoS({\cal G}_2)\geq 2.7558-\delta$.
\end{theorem}
\begin{proof}
Consider an instance with $n=n_1+n_2$ players divided into two sets $P_1$ and $P_2$ with $|P_1|=n_1$ and $|P_2|=n_2$. Each player $i\in P_1$ has two strategies $s_i^K$ and $s_i^O$, while all players in $P_2$ have the same strategy $\overline{s}$.

There are three types of resources:
\begin{itemize}
\item[$\bullet$] $n_1$ resources $r_i$, $i\in [n_1]$, each with latency function $\ell_{r_i}(x)=r x^d$. Resource $r_i$ belongs only to $s_i^O$;
\item[$\bullet$] $n_1(n_1-1)$ resources $r'_{i,j}$, $i,j\in [n_1]$ with $i\neq j$, each with latency function $\ell_{r'_{ij}}(x)=r' x^d$. Resource $r'_{ij}$ belongs only to $s_i^K$ and to $s_j^O$;
\item[$\bullet$] one resource $r''$ with latency function $\ell_{r''}(x)=x^d$. Resource $r''$ belongs to $s_i^K$ for each $i\in [n_1]$ and to $\overline{s}$;
\end{itemize}
The cost of each player $i\in P_1$ adopting strategy $s_i^K$ when there are exactly $k$ players in $P_1$ adopting the strategy played in $K$ (and thus there are $n_1-k$ players in $P_1$ adopting the strategy played in $O$) is $cost_K(k)=(4n_1-3k-1)r'+(n_2+k)^2$ when $d=2$ and it is $cost_K(k)=(8n_1-7k-1)r'+(n_2+k)^3$ when $d=3$. Similarly, the cost of each player $i\in P_1$ adopting strategy $s_i^O$ when there are exactly $k$ players in $P_1$ adopting the strategy played in $K$ is $cost_O(k)=r+(n_1+3k-1)r'$ when $d=2$ and it is $cost_O(k)=r+(n_1+7k-1)r'$ when $d=3$.

We now want to select the parameters $r$ and $r'$ so that $K$ is the unique PNE of the game. This is true if, for any $k\in [n_1]$, it holds $cost_O(k-1)>cost_K(k).$ Such a condition is always verified for the following values of $r$ and $r'$:
\begin{itemize}
\item[$\bullet$] $r=\frac{2n_2^2+(n_1+1)(n_1+2n_2)}{2}+\gamma$ and $r'=\frac{n_1+2n_2}{6}$, when $d=2$,
\item[$\bullet$] $r=\frac{2n_2^3+(n_1+1)(n_1^2+3n_1 n_2+3n_2^2)}{2}+\gamma$ and $r'=\frac{n_1^2+3n_1 n_2+3n_2^2}{14}$, when $d=3$,
\end{itemize}
where $\gamma$ is an arbitrarily small positive value.

Next step is to select $n_1$ and $n_2$ so as to maximize the ratio $\frac{\CSUM(K)}{\CSUM(O)}=\frac{r' n_1(n_1-1)+(n_1+n_2)^{d+1}}{r n_1+r' n_1(n_1-1)+n_2^{d+1}}$ for $d=2,3$. By choosing $n_1=1.5595 n_2$ when $d=2$ and $n_1=1.0988 n_2$ when $d=3$ and letting $n_2$ go to infinity, we obtain the claim.\qed
\end{proof}

\subsection{Bounding the Price of Anarchy for Social Function Max}\label{secmax}
In this section we show how the primal-dual technique can be adapted also to the case in which the social function is the maximum of the players' payoffs. For the sake of brevity, we consider only the problem of bounding the PoA in linear congestion games with unweighted players. In order to deal with the maximum social function, we assume, without loss of generality, that player $n$ is the one paying the highest cost in $K$ and impose that, in $O$ no player pays more than one.

Thus, $LP(K,O)$ has the following form.

\begin{displaymath}
\begin{array}{rlr}
maximize & k\\
subject\ to\\
\displaystyle\sum_{e\in s_i^K}\left(\alpha_e K_e\right)-\sum_{e\in s_i^O}\left(\alpha_e (K_e+1)\right) & \leq 0, & \forall i\in [n]\\
\displaystyle\sum_{e\in s_i^K}\left(\alpha_e K_e\right) & \leq k, & \forall i\in [n-1]\\
\displaystyle\sum_{e\in s_n^K}\left(\alpha_e K_e\right) & = k\\
\displaystyle\sum_{e\in s_i^O}\left(\alpha_e O_e\right) & \leq 1, & \forall i\in [n]\\
\alpha_e & \geq 0, & \forall e\in E
\end{array}
\end{displaymath}
$DLP(K,O)$ is as follows.
\begin{displaymath}
\begin{array}{rlr}
minimize & \displaystyle\sum_{i\in [n]} z_i\\
subject\ to\\
\displaystyle\sum_{i:e\in s_i^K}\left(K_e(x_i+y_i)\right)-\sum_{i:e\in s_i^O}\left(x_i (K_e+1)-z_i O_e\right) & \geq 0, & \forall e\in E\\
\displaystyle\sum_{i\in [n]} y_i & \geq -1\\
x_i & \geq 0, & \forall i\in [n]\\
y_i & \geq 0, & \forall i\in [n-1]\\
z_i & \geq 0, & \forall i\in [n]
\end{array}
\end{displaymath}

We easily reobtain the upper bound on the PoA proven in \cite{CK05}.

\begin{theorem}
PoA$({\cal G})={\cal O}(\sqrt{n})$ for any $\cal G$ with unweighted players.
\end{theorem}
\begin{proof}
By choosing $x_i=\frac{1}{\sqrt{n}}$, $y_i=0$ and $z_i=\frac{2}{\sqrt{n}}$ for each $i\in [n-1]$ and $x_n=1$, $y_n=-1$ and $z_n=2\sqrt{n}$, the dual inequalities may assume different forms depending of which of the following four situation occurs.

\begin{itemize}
\item[$\bullet$] $e\notin s_n^K\wedge e\notin s_n^O$. In this case, we have $K_e^2-K_e O_e-O_e+2O_e^2\geq 0$. Easy calculations show that this is always verified for any pair of non-negative integers $(K_e,O_e)$.
\item[$\bullet$] $e\notin s_n^K\wedge e\in s_n^O$. In this case, we have $K_e^2-K_e(\sqrt{n}+O_e-1)+2nO_e-\sqrt{n}+(O_e-1)(2O_e-1)\geq 0$. Easy calculations show that this is always verified for any pair of non-negative integers $(K_e,O_e)$ with $O_e\geq 1$.
\item[$\bullet$] $e\in s_n^K\wedge e\notin s_n^O$. In this case, we have $K_e^2-K_e(O_e+1)-O_e+2O_e^2\geq 0$. Easy calculations show that this is always verified for any pair of non-negative integers $(K_e,O_e)$ with $K_e\geq 1$.
\item[$\bullet$] $e\in s_n^K\wedge e\in s_n^O$. In this case, we have $K_e^2-K_e (\sqrt{n}+O_e) + 2nO_e-\sqrt{n}+(O_e-1)(2O_e-1)\geq 0$. Easy calculations show that this is always verified for any pair of non-negative integers $(K_e,O_e)$ with $K_e\geq 1$ and $O_e\geq 1$.
\end{itemize}

The claim follows since $\sum_{i=1}^n y_i=-1$ and $\sum_{i=1}^n z_i\leq 4\sqrt{n}$.\qed
\end{proof}

\subsection{Resource Allocation Games with Fair Cost Sharing}\label{secfair}
In this section we show how the primal-dual technique can be adapted also to the study of the efficiency of PNE in resource allocation games with fair cost sharing. We briefly recall that in these games players choose subsets of resources. Each resource $e\in E$ has an associated cost $c_e$ and the cost of each resource is equally shared among all players using it in a given strategy profile. For a strategy profile $S$, let $p_e(S)\in\{0;1\}$ be a boolean variable such that $p_e(S)=1$ if there exists at least a player using resource $e$ in $S$ and $p_e(S)=0$ otherwise.

For the PoS of these games, $LP(K,O)$ has the following form.

\begin{displaymath}
\begin{array}{rlr}
maximize & \displaystyle\sum_{e\in E}(c_e p_e(K))\\
subject\ to\\
\displaystyle\sum_{e\in s_i^K}\left(\frac{c_e}{n_e(K)}\right)-\sum_{e\in s_i^O}\left(\frac{c_e}{n_e(K)+1}\right) & \leq 0, & \forall i\in [n]\\
\displaystyle\sum_{e\in E}\left(c_e(H_{K_e}-H_{O_e})\right) & \leq 0\\
\displaystyle\sum_{e\in E}\left(c_e p_e(O)\right) & = 1, \\
c_e & \geq 0, & \forall e\in E
\end{array}
\end{displaymath}
$DLP(K,O)$ is as follows.
\begin{displaymath}
\begin{array}{rlr}
minimize & \gamma\\
subject\ to\\
\displaystyle\sum_{i:e\in s_i^K}\left(\frac{y_i}{n_e(K)}+z H_{K_e}\right)\\
-\displaystyle\sum_{i:e\in s_i^O}\left(\frac{y_i}{n_e(K)+1}+z H_{O_e}\right)+\gamma p_e(O) & \geq p_e(K), & \forall e\in E\\
y_i & \geq 0, & \forall i\in [n]\\
z & \geq 0
\end{array}
\end{displaymath}

\begin{theorem}
For any resource selection game with fair cost sharing $\cal G$, PoS$({\cal G})\leq H_n$.
\end{theorem}
\begin{proof}
Set $y_i=0$ for any $i\in [n]$, $z=1$ and $\gamma=H_n$. The claim easily follows from the fact that, for any $e\in E$, $H_{K_e}\geq p_e(K)$ and $H_n p_e(O)\geq H_{O_e}$.\qed
\end{proof}

Clearly, obtaining such a result is quite straightforward. We illustrate this further application of the primal-dual technique with the main purpose of showing that it can be fruitfully used also to contexts other than weighted congestion games.


\begin{thebibliography}{99}

\bibitem{ADGMS06}
S. Aland, D. Dumrauf, M. Gairing, B. Monien, and F. Schoppmann. Exact Price of Anarchy for Polynomial Congestion Games. In {\em Proceedings of the 23rd Annual Symposium on Theoretical Aspects of Computer Science (STACS)}, LNCS 3884, Springer, pp. 218-229. 2006.

\bibitem{ADKTWR04}
E. Anshelevich, A. Dasgupta, J. M. Kleinberg, E. Tardos, T. Wexler, and T. Roughgarden. The Price of Stability for Network Design with Fair Cost Allocation. {\em SIAM Journal on Computing}, 38(4), pp. 1602-1623, 2008.

\bibitem{AAE05}
B. Awerbuch, Y. Azar, and A. Epstein. The Price of Routing
Unsplittable Flow. In {\em Proceedings of the 37th Annual ACM
Symposium on Theory of Computing (STOC)}, ACM Press, pp. 57-66,
2005.

\bibitem{BGR10}
K. Bhawalkar, M. Gairing, and T. Roughgarden. Weighted Congestion Games: Price of Anarchy, Universal Worst-Case Examples, and Tightness. In {\em Proceedings of the 18th Annual European Symposium on Algorithms (ESA)}, LNCS 6347, Springer, pp. 17-28, 2010.

\bibitem{BFFM09}
V. Bil\`o, A. Fanelli, M. Flammini, and L.
Moscardelli. Performances of One-Round Walks in Linear Congestion Games. In
{\em Proceedings of the 2nd International Symposium on Algorithmic Game Theory (SAGT)}, LNCS 5814, Springer,
pp. 311-322, 2009.

\bibitem{CFKKM06}
I. Caragiannis, M. Flammini, C. Kaklamanis, P. Kanellopoulos, and L.
Moscardelli. Tight Bounds for Selfish and Greedy Load Balancing. In
{\em Proceedings of the 33rd International Colloquium on Automata,
Languages and Programming (ICALP)}, LNCS 4051, Springer,
pp. 311-322, 2006.

\bibitem{CK05}
G. Christodoulou and E. Koutsoupias. The Price of Anarchy of Finite
Congestion Games. In {\em Proceedings of the 37th Annual ACM
Symposium on Theory of Computing (STOC)}, ACM Press, pp. 67-73,
2005.

\bibitem{CK05b}
G. Christodoulou and E. Koutsoupias. On the Price of Anarchy and
Stability of Correlated Equilibria of Linear Congestion Games. In
{\em Proceedings of the 13th Annual European Symposium on Algorithms
(ESA)}, LNCS 3669, Springer, pp. 59-70, 2005.

\bibitem{CKS10}
G. Christodoulou, E. Koutsoupias, and P. G. Spirakis. On the Performance of Approximate Equilibria
in Congestion Games. {\em Algorithmica}, to appear.

\bibitem{CMS06}
G. Christodoulou, V. S. Mirrokni, and A. Sidiropoulos.
Convergence and Approximation in Potential Games.
In {\em Proceedings of the 23rd International Symposium on Theoretical Aspects of Computer Science (STACS)}, LNCS 3884, Springer, pp. 349-360, 2006.

\bibitem{FKS05}
D. Fotakis, S. C. Kontogiannis and P. G. Spirakis. Symmetry in Network Congestion Games: Pure Equilibria and Anarchy Cost. In {\em Proceedings of the 3rd Workshop on Approximation and Online Algorithms (WAOA)}, LNCS 3879, Springer, pp. 161-175, 2005.

\bibitem{HK10}
T. Harks and M. Klimm. On the Existence of Pure Nash Equilibria in Weighted Congestion Games. In {\em Proceedings of the 37th International Colloquium on Automata, Languages and Programming (ICALP)}, LNCS 6198, Springer, pp. 79-89, 2010.

\bibitem{HKM09}
T. Harks, M. Klim and R.H. M{\"o}hring. Characterizing the Existence of Potential Functions in Weighted Congestion Games. In {\em Proceedings of the 2nd International Symposium on Algorithmic Game Theory (SAGT)}, LNCS 5814, Springer, pp. 97-108, 2009.

\bibitem{KP99}
E. Koutsoupias and C. Papadimitriou. Worst-Case Equilibria. In {\em Proceedings of the 16th International Symposium on Theoretical Aspects of Computer Science (STACS)}, LNCS 1563, Springer, pp. 404-413, 1999.

\bibitem{MV04}
V. S. Mirrokni and A. Vetta. Convergence Issues in Competitive Games. In {\em Proceedings of the 7th International Workshop on Approximation Algorithms for Combinatorial Optimization Problems (APPROX)}, LNCS 3122, Springer, pp. 183-194, 2004.

\bibitem{NR10}
U. Nadav and T. Roughgarden. The Limits of Smoothness: A Primal-Dual Framework for Price of Anarchy Bounds. In {\em Proceedings of the 6th International Workshop on Internet and Network Economics (WINE)}, LNCS 6484, Springer, pp. 319-326, 2010.

\bibitem{PS06}
P. N. Panagopoulou and P. G. Spirakis. Algorithms for Pure Nash Equilibria in Weighted Congestion Games. {\em Journal of Experimental Algorithmics},
11:(2.7), 2006.

\bibitem{R09}
T. Roughgarden. Intrinsic Robustness of the Price of Anarchy. In {\em Proceedings of the 41st Annual ACM Symposium on Theory of Computing (STOC)}, ACM Press, pp. 513-522, 2009.

\bibitem{R73}
R. W. Rosenthal. A Class of Games Possessing Pure-Strategy Nash Equilibria. {\em International Journal of Game Theory}, 2:65-67, 1973.

\end{thebibliography}
\end{document}